\documentclass[conference,10pt]{IEEEtran}
\IEEEoverridecommandlockouts
\usepackage[letterpaper,
left=1in,right=1in,top=1.91cm,bottom=2.55cm]{geometry}

\usepackage{latexsym}
\usepackage{graphicx}
\usepackage{array}
\usepackage{amsmath}
\usepackage{amsfonts}
\usepackage{amssymb}
\usepackage{amsthm}
\usepackage{algpseudocode}
\usepackage{algorithm}
\usepackage{subfigure}

\newtheorem{thm}{Theorem}

\newtheorem{prop}{Proposition}
\newtheorem{defn}{Definition}

\newtheorem{cor}{Corollary}

\newcommand{\bx} {\boldsymbol{x}}
\newcommand{\by} {\boldsymbol{y}}
\newcommand{\bc} {\boldsymbol{c}}
\newcommand{\ba} {\boldsymbol{a}}
\newcommand{\cnv} {\mathbf{conv}}

\newcommand{\gl}{\lambda}

\def\bal#1\eal{\begin{align}#1\end{align}}
\newcommand{\bp} {\begin{proof}}
\newcommand{\ep} {\end{proof}}

\newcommand{{\bRF}} {\right\}}

\begin{document}

\title{Optimal Location of Cellular Base Station via Convex Optimization}

	
	

\author{\IEEEauthorblockN{Elham Kalantari\IEEEauthorrefmark{1},
		Sergey Loyka\IEEEauthorrefmark{1},
		Halim Yanikomeroglu\IEEEauthorrefmark{2}, and 
		Abbas Yongacoglu\IEEEauthorrefmark{1}}
	\IEEEauthorblockA{\IEEEauthorrefmark{1}School of Electrical Engineering and Computer Science\\
		University of Ottawa, Ottawa, ON, Canada, Email: \{ekala011, sergey.loyka, yongac\}@uottawa.ca}
	\IEEEauthorblockA{\IEEEauthorrefmark{2}Department of Systems and Computer Engineering\\
		Carleton University, Ottawa, ON, Canada, Email: halim@sce.carleton.ca}}

\maketitle


\begin{abstract}
    An optimal base station (BS) location depends on the traffic (user) distribution, propagation pathloss and many system parameters, which renders its analytical study difficult so that numerical algorithms are widely used instead. In this paper, the problem is studied analytically. First, it is formulated as a convex optimization problem to minimize the total BS transmit power subject to quality-of-service (QoS) constraints, which also account for fairness among users. Due to its convex nature, Karush-Kuhn-Tucker (KKT) conditions are used to characterize a globally-optimum location as a convex combination of user locations, where convex weights depend on user parameters, pathloss exponent and overall geometry of the problem. Based on this characterization, a number of closed-form solutions are obtained. In particular, the optimum BS location is the mean of user locations in the case of free-space propagation and identical user parameters. If the user set is symmetric (as defined in the paper), the optimal BS location is independent of pathloss exponent, which is not the case in general. The analytical results show the impact of propagation conditions as well as system and user parameters on optimal BS location and can be used to develop design guidelines.
\end{abstract}


\section{Introduction}

The problem of base station (BS) location in cellular networks has been extensively studied in the existing literature, see e.g. \cite{Amaldi2003}-\cite{Kalantari2017pimrc}. A number of optimization algorithms have been proposed to attack this problem numerically, taking into account a number of practically-important parameters and limitations. Many of the proposed algorithms use a pre-selected finite list of candidate sites where the BS could potentially be located and look for the ones that optimize some objective function amongst that list \cite{Amaldi2003}-\cite{Zhao2017}. The considered problems are formulated as mixed integer programming or combinatorial optimization and the methods to solve them include simulated annealing, Tabu search, simplex method and branch and bound algorithm, etc. While these approaches can be useful in practice, their common feature is that the considered problems are NP-hard (i.e. the numerical complexity grows very fast with the problem size), and convergence of algorithms to a global optimum cannot be guaranteed. A different approach is adopted in \cite{Sherali1996}, where the weighted sum pathloss (to all users) was minimized without any pre-selected BS locations. Several numerical algorithms for local optimization were used, such as Hooke-Jeeves’, quasi-Newton and conjugate gradient search. However, the cost function was introduced in an ad-hoc manner, without any explicit link to typical system-level performance indicators (e.g. total power or energy efficiency), and the resulting optimization problem was not convex.

While the above algorithms are useful from the practical perspective, they have a number of limitations at the fundamental level. In all these algorithms, convergence to a global optimum cannot be guaranteed either due to non-convexity of underlying optimization problems or inherent limitations (approximations) of the algorithms. Furthermore, a gap to a globally-optimal solution is not known or bounded either. Due to the numerical nature of the algorithms, very limited or no insights are available. No closed-form solutions to the considered problems are known either.

In this paper, we adopt a different approach. Optimal BS location is modeled as a convex optimization problem to minimize the total BS transmit power, subject to per-user quality-of-service (QoS) constraints, which also account for fairness among users. Due to the convex nature of our formulation, the respective KKT conditions are sufficient for global optimality, from which a number of closed-form solutions can be obtained and numerical algorithms can also be built with a guaranteed convergence to a global optimum (using e.g. the barrier method) \cite{Boyd2004}. The emphasis of this paper is on the analysis, closed-form solutions and insights they facilitate, rather than on numerical algorithms.
The system model is introduced in Section \ref{sec.model}. This model may represent actual users in a cellular system with their rate requirements as well as expected user distributions (e.g. in business or apartment buildings, shopping centers and other social attractors); expected traffic demands in different locations can also be represented in this way via virtual users (whose locations and number are representative of the expected traffic demand). The considered model and approach are general enough to include any user rate that is a monotonically-increasing function of the SNR and hence can include fading, in addition to the average pathloss, as well as non-uniform user distributions. It also applies to 3D scenarios, typical of unmanned aerial vehicles (UAVs) or other mobile BSs \cite{Kalantari2016}-\cite{Cao-18}. While our model is not as general as some other models in the literature, it makes the problem analytically-tractable and a number of novel closed-form solutions and properties follow.

An optimal BS location subject to QoS constraints is formulated as a convex optimization problem to minimize the total BS transmit power, which is a key to the further development. Based on this formulation, an optimal BS location is characterized as a convex combination of user locations in the general 3D case, where the convex weights depend on user bandwidth and rate demands, some system and propagation parameters, and overall geometry of the problem, see Theorem \ref{thm.c*}. This characterization is subsequently used to obtain a number of explicit closed-form solutions for an optimal BS location (to the best of our knowledge, for the first time). In the case of free-space propagation, the optimal BS location is  the weighted mean of user locations (where the weights are determined by system parameters). If, in addition, the users have identical parameters (rate and bandwidth), the problem further  reduces to the well-known facility location problem (in squared Euclidean norm metric) and the optimal BS location is the mean of user locations. Our novel contribution here is that this BS location also minimizes its total transmit power under free-space propagation and identical system parameters of the users (but not otherwise in general).

We further show that this result also applies to other propagation environments (with other pathloss exponents), provided that the set of users is symmetric in a certain way, see Definitions 2 and 3. Hence, this result is more general than originally expected. Furthermore, the optimal BS location is also independent of pathloss exponent $\nu$ in this case while $\nu$ has a profound impact on it for asymmetric user sets. In the case of large pathloss exponent, the optimal BS location is determined by the most distant users. An unusual property is observed whereby an optimal BS location is not necessarily unique: while it is always unique when the pathloss exponent $\nu >1$, this is not the case with $\nu=1$.
These results are further extended to include additional location constraints (due to e.g. existing infrastructure) as well as elevated BS scenarios (e.g. UAV-BS), see Theorems \ref{thm.c*.h} and \ref{thm.c*.loc}.

The analytical results above, i.e. an optimum BS location (to minimize its total transmit power), its geometric properties as well as the impact of pathloss exponent and user distribution on this location are, to the best of our knowledge, novel and cannot be found in the existing literature. They render insights unavailable from purely numerical studies, which can be subsequently applied to obtain design guidelines for more complicated scenarios, for which no analytical solutions are known.

It is worthwhile to note that, in the special case of pathloss exponent $\nu=1$ and identical user parameters, the problem considered here reduces to the celebrated "Fermat-Weber" problem \cite{Drezner2002}, which is to find a point that minimizes the sum of its distances to a set of given points, and for which no closed-form solution is known to this day in the general case. To quote \cite{Drezner2002}, \textit{"The Weber problem ... has a long and convoluted history. Many players, from many fields of study, stepped on its stage, and some of them stumbled. The problem seems disarmingly simple, but is so rich in possibilities and traps that it has generated an enormous literature dating back to the seventeenth century, and continues to do so."}

\section{System Model and Problem Formulation}
\label{sec.model}

Let us consider a BS serving $N$ users located at $\bx_k,\ k = 1,.., N$, via some form of orthogonal multiple-access technique (e.g. FDMA). We require user rates $R_k$ to be monotonically-increasing functions of the SNR, e.g.
\bal
\label{eq.Rk}
R_k = \Delta f_k \log(1+\gamma_k/\Gamma_k),
\eal
where $\Delta f_k$ and $\gamma_k= P_{rk}/\sigma_{0k}^2$ are the bandwidth and the SNR of user $k$, the channel is frequency-flat with AWGN noise of power $\sigma_{0k}^2$ and $P_{rk}$ is the signal power received by user $k$; $\Gamma_k \ge 1$ is the SNR gap to the capacity of user $k$ \cite{Forney-98}. When efficient (capacity-approaching) codes are used for each user, $\Gamma_k \rightarrow 1$. The received power $P_{rk}$ is related to the transmit power $P_k$ allocated by the BS to user $k$ via the pathloss model, see e.g. \cite{Rappaport2001},
\bal
P_{rk} = \alpha_k P_k/d_k^{\nu_k},
\eal
where $d_k = |\bc - \bx_k|$  is the distance between the BS located at $\bc$ and user $k$ located at $\bx_k$, $|\bx|$ is the Euclidean norm (length) of vector $\bx$, $\nu_k$ is the pathloss exponent, and $\alpha_k$ is a constant related to the propagation environment, which is independent of distance but may depend on frequency. For example, in the case of free-space propagation environment, e.g. when line-of-sight (LoS) path is dominant, $\nu_k=2$ and $\alpha_k = (\lambda_k/(4\pi))^2$, where $\lambda_k$ is the wavelength of user $k$, while for the  2-ray ground reflection model $\nu_k=4$ and $\alpha_k = h_t^2 h_{rk}^2$, where $h_t, h_{rk}$ are the BS and user $k$ antenna heights \cite{Rappaport2001}, all in the far-field. 

We assume that the BS knows the pathloss to each user (or, equivalently, its SNR). To satisfy QoS requirements, each user rate must not be less than its target rate $R_{0k}$: $R_k \ge R_{0k}$. To achieve this objective in an energy-efficient way, the operator selects BS location $\bc$ in an optimal way to minimize its total transmit power $P_T = \sum_k P_{k}$ subject to the QoS constraints as follows:
\begin{equation}
\label{eq.P1}
\begin{aligned}
\text{\textrm{(P1)}}~~& \underset{\{P_{k}\}, \bc}{\text{min}}
& & \sum_k P_{k}\  \ \text{s.t.}\ \ R_k \ge R_{0k},
\end{aligned}
\end{equation}
where the optimization variables are BS location $\bc$ as well as per-user powers $\{P_k\}$, so that the BS performs optimal per-user power allocation as well. The rate constraints $R_k \ge R_{0k}$ also ensure fairness among users. Noting from \eqref{eq.Rk} that the constraint $R_k \ge R_{0k}$ is equivalent to $\gamma_k \ge \gamma_{0k}= (2^{R_{0k}/\Delta f_k}-1)\Gamma_k$, the problem (P1) can be re-formulated as follows:
\begin{equation}
\label{eq.P2}
\begin{aligned}
\text{\textrm{(P2)}}~~& \underset{\{P_{k}\}, \bc}{\text{min}}& \sum_k P_{k}\ \ \text{s.t.}\ \ P_{k} \ge  \beta_k |\bc - \bx_k|^{\nu_k},
\end{aligned}
\end{equation}
where $\beta_k = \gamma_{0k}\sigma_{0k}^2/ \alpha_k$. Note that $\sigma_{0k}^2$ may also include interference power as a part of it. We further note that problem (P1) and hence (P2) can also accommodate any rate model that is a monotonically-increasing function of the SNR $R_k(\gamma_k)$, not only that in \eqref{eq.Rk}, so that the condition $R_k \ge R_{0k}$ is equivalent to $\gamma_k \ge \gamma_{0k}$ with properly-selected $\gamma_{0k}= R^{-1}_k(R_{0k})$. This generalized model can also include fading, where $R_k$ and $\gamma_k$ are interpreted as the average (ergodic) rate and SNR respectively. It should be emphasized that the problem formulation (P2) is based on power/energy minimization, unlike some other  formulations in the literature (e.g. \cite{Sherali1996}\cite{Drezner2002}) where the objective (cost) function is introduced in an ad-hoc way. The restriction to a single BS is necessary to make the problem analytically tractable (which seems to be out of reach otherwise). However, minimizing the BS power in one cell as in \eqref{eq.P1}, \eqref{eq.P2} will also reduce the amount of inter-cell interference it generates to other cells under frequency re-use.

\section{Optimal BS Location and Power Allocation}
\label{sec.Optimal BS Location}

To the best of our knowledge, no analytical solution is available in the literature to either (P1) or (P2) in the general case (even though the setting is limited to a single BS). Therefore, we present  next a general characterization of an optimal BS location according to (P2) (see Appendix for a proof), from which a number of closed-form solutions follow.

\begin{thm}
\label{thm.c*}
An optimal BS location $\bc^*$ for (P2) in \eqref{eq.P2} can be expressed as a convex combination of user locations $\{\bx_k\}$:
\begin{equation}
\label{eq.c*}
\bc^* = \sum_k \theta_k \bx_k,\ \theta_k = \frac{\beta_{k}\nu_k|\bc^* - \bx_k|^{\nu_k - 2}}{\sum_k \beta_{k}\nu_k|\bc^* - \bx_k|^{\nu_k -2}}
\end{equation}
if either (i) $\nu_k \ge 2$ or/and (ii) $\bc^* \neq \bx_k$ and $\nu_k \ge 1$. Transmission with the least per-user power is optimal: $P_{k}^*  = \beta_k |\bc^* - \bx_k|^{\nu_k}$.
\end{thm}

Next, we explore some properties of an optimal BS location.

\begin{prop}
\label{prop.unique.c*}
When $\nu_k > 1$ for some $k$, an optimal BS location is unique. This is not necessarily the case if $\nu_k=1$ for all $k$.
\begin{proof}
Observe that
(P2) is equivalent to $\min _{\bc} \sum_k \beta_k |\bc - \bx_k|^{\nu_k}$,
since transmitting with the least per-user power is optimal, and that the objective here is strictly convex if $\nu_k >1$ for some $k$, so that the solution is unique \cite{Boyd2004}. Non-uniqueness for $\nu_k =1$ can be shown via examples, see Proposition \ref{prop.nu_1}.
\end{proof}
\end{prop}

To obtain some insights, we need the following definition \cite{Boyd2004}, from which Corollary \ref{cor.conv} follows.

\begin{defn}
\label{defn.conv}
Let $\{\by_k\}$ be a set of points. Its convex hull $\cnv\{\by_k\}$ is the set of all convex combinations of the points in $\{\by_k\}$:
\begin{equation}
\cnv\{\by_k\} = \left\{\sum_k q_k\by_k: \ q_k\ge 0,\ \sum_k q_k=1\right\}
\end{equation}
\end{defn}

\begin{cor}
\label{cor.conv}
The optimal BS location $\bc^*$ in \eqref{eq.c*} is in the convex hull of all user locations:
\bal
\bc^* \in \cnv\{\bx_k\}
\eal
\end{cor}
\begin{proof}
Notice from \eqref{eq.c*} that $0 \le \theta_k  \le 1,\ \sum_k \theta_k =1$, and then apply Definition \ref{defn.conv}.
\end{proof}

The above Corollary implies that the search of $\bc^*$ can always be confined to $\cnv\{\bx_k\}$, without loss of optimality. For example, if all users are located on a line or in a convex building, the optimal BS is also on this line or in this building.
We obtain below a number of explicit closed-form solutions for $\bc^*$ in some special cases.

\subsection{Free-space propagation}

The first important special case is that of free-space propagation, where $\nu_k=2$. In practice, $\nu_k$ is close to 2 when propagation is close to free space, i.e. most of the 1st Fresnel zone is free of obstructions \cite{Rappaport2001}. This is also the case in a multipath channel when multipath components are much weaker than LoS; therefore, LoS dominates and the propagation becomes almost the same as in free space. $\nu_k$ is close to 2 in many indoor environments when LoS is present \cite{Rappaport2001} and $\nu_k=2$ appears often in the 3GPP LTE propagation models. Using Theorem \ref{thm.c*}, $\bc^*$ can be expressed as follows in this case.

\begin{cor}
\label{cor.nu2}
If $\nu_k=2$ for all $k$, the optimal BS location $\bc^*$ is a weighted mean of the user locations:
\begin{equation}
\label{eq.c.nu2}
\bc^*  = \sum_k \theta_k \bx_k, \theta_k = \frac{\beta_{k}}{\sum_i \beta_{i}}.
\end{equation}
\end{cor}
\begin{proof}
Use \eqref{eq.c*} with $\nu_k=2$.
\end{proof}
Note that \eqref{eq.c.nu2} is an explicit closed-form solution, since $\theta_k$ are now independent of $\bc^*$. It follows that users with larger $\beta_k$, i.e.  those requiring higher rates, contribute more to $\bc^*$ so that as $\beta_k$ increases, $\bc^*$ moves closer to $\bx_k$. In the limiting case of $\beta_1 > 0,\ \beta_i=0, i \neq 1$, the optimal location $\bc^*=\bx_1$.

Further simplification is possible when all users require the same rate and have the same system settings, so that $\beta_k = \beta\ \forall k$. In this case, \eqref{eq.c.nu2} reduces to the mean value of the users' locations - a result well-known in the facility location literature under the Euclidean norm squared criterion \cite{Boyd2004}. Our novel contribution here is that this BS location minimizes its total transmit power under free-space propagation and identical system parameters of the users (but not otherwise in general).

\subsection{Large pathloss exponent}

To obtain further insights, we consider the limiting case of large pathloss exponent $\nu_k \to\infty$, which serves as an approximation to large but finite $\nu_k$. To simplify the discussion, we further assume that all users have identical parameters so that $\beta_k =\beta\ \forall k$.

\begin{prop}
\label{prop.nu.inf}
If $\nu_k \to\infty$, the optimal BS location is the mean of most distant user locations.
\end{prop}
\begin{proof}
Use \eqref{eq.c*} and take the limit $\nu_k \to \infty$.
\end{proof}

Hence, for large pathloss exponent, it is the most distant users who determine the optimal BS location, while nearby users contribute little. Finding most distant users in a set can be expressed as a geometric (and convex) problem of finding the smallest enclosing sphere where optimization variables are the sphere center $\bc$ and its radius $r$:
\begin{equation}
\begin{aligned}
& \underset{r,\bc}{\text{min}}\ r\ \text{s.t.} & |\bc - \bx_k| \le r\ \forall k.
\end{aligned}
\label{cvx-infinity}
\end{equation}

\subsection{Symmetric sets of users}

To obtain closed-form solutions for $\bc^*$ beyond those above, we consider now the scenarios where user location sets possesses some symmetry properties. This should also approximate (due to the continuity of the problem in user locations) the scenarios where users are nearly-symmetric. We will need the following definitions of symmetric sets.

\begin{defn}
    Let $\Omega_l = \{ \bx_k: k \in I_l \}$ be a set of $|I_l|$ points (users), where $I_l$ is an index set and $|I_l|$ is its cardinality. The set $\Omega_l$ is called elementary symmetric if the distance between its center $\ba_l = |I_l|^{-1} \sum _{k \in I_l} \bx_{k}$ and any of its points is the same, i.e. $|\ba_l - \bx_k | = d_l\ \forall k \in I_l$.
\end{defn}

\begin{defn}
    Set $\Omega$ is symmetric if it is a union of disjoint elementary symmetric sets with the same centers, i.e. $\Omega = \cup_l \Omega_{l}$ and $\ba_l = \ba\ \forall l$.
\end{defn}

While an elementary symmetric set is also symmetric, the converse is not true in general, i.e. a symmetric set does not need to be elementary symmetric, as Fig. \ref{symmetric} illustrates, so the former is more general than the latter. Equipped with these notions of symmetry, we are now able to obtain the optimal BS location in a closed form.

\begin{figure}[t]
    \begin{center}
        \includegraphics[width=2.2in]{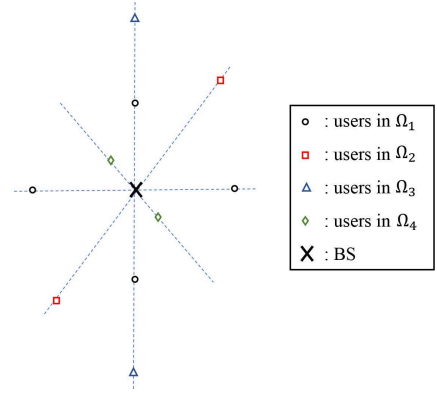}\\
    \end{center}
    \caption{The union of 4 elementary symmetric sets $\Omega_1 .. \Omega_4$ with the same center is symmetric; the optimal BS location, for any pathloss exponent $\nu$, is its (common) center.}
    \label{symmetric}
\end{figure}

\begin{prop}
    \label{prop.symmetric}
     Let the set $\Omega$ of user locations be symmetric, i.e. $\Omega = \cup_l \Omega_{l}$, where $\Omega_{l}$ are disjoint and elementary-symmetric, $\nu_k=\nu_l$ for any $k\in I_l$, and $\beta_k = \beta\ \forall k$. Then, for any pathloss exponents $\nu_k > 1$ for all $k$, the optimal BS location is its center $\ba$, i.e. the mean of the users' locations,
     \begin{equation}
    \label{eq.c*.sym}
    \bc^* = \ba = \overline{\bx} =\frac{1}{N}\sum_k \bx_k.
\end{equation}
\end{prop}
\begin{proof}
Using \eqref{eq.c*} and exploiting the symmetry properties, along with the convexity of the objective functions, results in \eqref{eq.c*.sym} after some manipulations, see \cite{Kalantari-20}.
\end{proof}

It should be emphasized that this result holds for any $\nu_k > 1$, not just for $\nu_k=2$, as in Corollary \ref{cor.nu2} (with $\beta_k=\beta$), so this result is more general in terms of $\nu_k$ but more restrictive in terms of user locations as symmetry is required here, unlike Corollary \ref{cor.nu2}. Note also that, unlike the general case,  the optimal BS location is independent of pathloss exponent $\nu_k$ as long as the user set is symmetric. This Proposition also implies that when new users are added to existing ones, the optimal BS location is not affected as long as new users do not disturb symmetry. It can be further shown that the BS location in \eqref{eq.c*.sym} also minimizes the amount of co-channel interference to the users of other cells provided they satisfy certain symmetry requirement.

\subsection{Collinear users}
\label{sec.collinear}

Let us consider the case where all users are located on a line. This is motivated by practical settings on highways, in tunnels, street canyons or corridors. Following  Corollary \ref{cor.conv}, an optimal BS location is also on the line, while its specific location depends on users' locations and pathloss exponent. We consider below the case of $\nu_k=1$ for all $k$ and demonstrate some unusual properties such as non-uniqueness of optimal BS location. Note that $\nu <2$ represents an environment more favorable for propagation than free space and it is possible in channels with guided wave structure, such as tunnels, corridors, street canyons \cite{Rappaport2001}.

\begin{prop}
\label{prop.nu_1}
Let all users to have the same system parameters, $\nu_k=1$, $\beta_k=\beta\ \forall k$, and be located on a line as represented by their scalar coordinates $x_k$, $k=1...N$; without loss of generality, set $x_1 \le x_2 \le \ldots \le x_N$. If $\nu_k=1$, an optimal BS location is a median of users' locations:
\begin{equation}
c^* = \begin{cases}
x_{(N+1)/2},\ N\ \mathrm{is\ odd}, \\
\mathrm{any}\ a \in [x_{N/2}, x_{N/2+1}],\ N\ \mathrm{is\ even}.
\end{cases}
\end{equation}
\end{prop}

\begin{figure}[t]
    \begin{center}
        \includegraphics[width=2.8in]{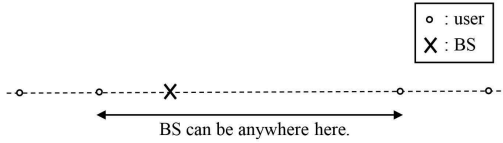}\\
    \end{center}
    \caption{If $\nu_k=1$ and the number of users is even, an optimal BS location is not unique: it can be anywhere between two middle-point users.}
    \label{nu=1fig}
\end{figure}

While this result is known in the facility location literature (under $L_1$ norm cost), our novel contribution here is that this BS location also minimizes its total transmit power under certain system and propagation settings (but not in general).

An illustration of Proposition \ref{prop.nu_1} is given in Fig. \ref{nu=1fig} when the number of users is even. Note that an optimal BS location is not unique in this case, which is ultimately due to the fact that $|x|$ is not strictly convex. However, if $\nu>1$, then it is always unique, according to Proposition \ref{prop.unique.c*}, since $|x|^{\nu}$ is strictly convex in this case. To see the impact of $\nu$, let us consider 3 special cases as shown in Fig. \ref{fig.collinear}:

1. For $\nu=1$, an optimal BS location is a median point (not unique -- can be anywhere between users 3 and 4).

2. For free-space propagation, $\nu=2$, the optimal BS location is the (unique) mean of the users' locations, according to Corollary \ref{cor.nu2}.

3. For asymptotically-large $\nu$, the optimal BS location is the mean of the most distant users' locations, according to Proposition  \ref{prop.nu.inf}, so that most distant users contribute most to optimal BS location in this case.

Thus, $\nu$ has a profound impact on optimal BS location for asymmetric user sets. This is in stark contrast with symmetric user sets (Proposition \ref{prop.symmetric}), where the optimal BS location is independent of $\nu$.

\begin{figure}[t]
    \centering
    \subfigure[$\nu=1$.]{
        \label{nu=1a}
        \includegraphics[width=2.7in]{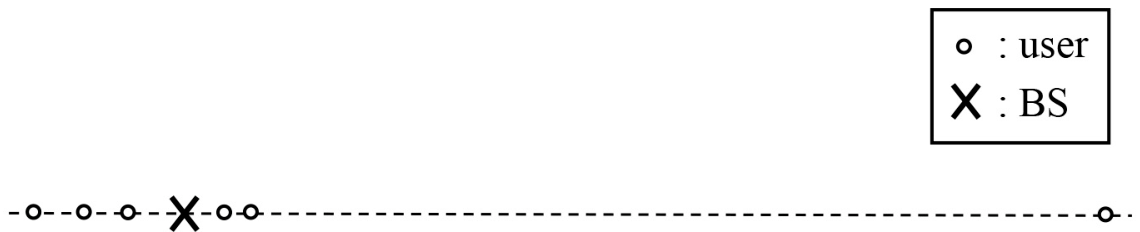}
    }
    \subfigure[$\nu=2$.]{
        \label{nu=two}
        \includegraphics[width=2.7in]{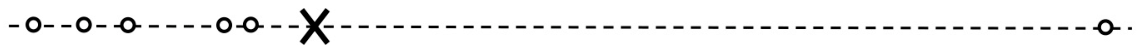}
    }
    \subfigure[$\nu\to\infty$.]{
        \label{nu=infinity}
        \includegraphics[width=2.7in]{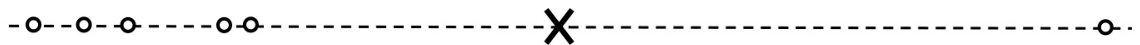}
    }
    \caption {Optimum BS locations for different pathloss exponents. For $\nu=1$, it is a median point, which is not unique (anywhere between users 3 and 4); for $\nu=2$ - the mean of the user locations; for $\nu \to \infty$ - the mean of the most distant users. As $\nu$ increases, the impact of the distant user on the right increases too.}
    \label{fig.collinear}
\end{figure}

\subsection{Elevated BS}
\label{sec.elevated}

In practice, BS is often located at some elevation above ground to provide clear LoS to most users hence improving coverage. This also includes scenarios with an airborne communication node (e.g. a drone).  To model this scenario, we consider a setting where all users are located on a (ground) plane with 2-D vector $\bx_k$ representing user $k$, while the BS is above the ground at a given height $h$ and $\bc$ is its 2-D location (projection) on the ground plane. The distance between the BS and user $k$ is therefore $\sqrt{|\bc-\bx_k|^2+h^2} = |\bc-\bx_k|_h$. Thus, the problem (P2) becomes
\begin{equation}
\label{eq.BS.elev}
\begin{aligned}
\underset{\{P_{k}\}, \bc}{\text{min}}\ & \sum_{k=1}^N P_{k} & \text{s.t.}\ P_{k} \ge  \beta_k |\bc-\bx_k|_h^{\nu_k}.
\end{aligned}
\end{equation}

The following Theorem characterizes its solutions.

\begin{thm}
\label{thm.c*.h}
Consider the elevated BS location problem in \eqref{eq.BS.elev} when $\nu_k\ge 1$. Its solution $\bc^*$ can be expressed as a convex combination of user locations $\{\bx_k\}$:
    \begin{equation}
    \label{eq.c*.h}
    \bc^* = \sum_k \theta_k \bx_k,\ \theta_k = \frac{\beta_{k}\nu_k|\bc^* -\bx_k|_h^{\nu_k-2}}{\sum_i \beta_{i}\nu_i|\bc^* -\bx_i|_h^{\nu_i-2}}.
    \end{equation}
\end{thm}
\begin{proof}

Follows from the proof of Theorem \ref{thm.c*}, see \cite{Kalantari-20}.
\end{proof}

\vspace*{-.71\baselineskip}
A number of properties/solutions pointed above also hold for the elevated BS problem in terms of its 2-D projected location
$\bc^*$. In particular, Corollaries \ref{cor.conv}-3, Propositions 2, 3, do hold for the elevated BS as well. Proposition 1 is strengthened as follows.

\begin{prop}
\label{prop.unique.c*.h}
The optimal elevated base station location is unique for any $\nu_k\ge 1$ if $h\neq 0$.
\end{prop}
\begin{proof}
Follows the steps of that of Proposition \ref{prop.unique.c*} by observing that $|\bx|_h^{\nu}$ is strictly convex for any $\nu \ge 1$ if $h\neq 0$.
\end{proof}

\subsection{Additional location constraints}
\label{sec.extra.const}

When locating a BS in practice, quite often there are some additional constraints due to existing infrastructure, such as a limited roof-top area available for a BS location. In such a case, the problem (P2) can be modified to include extra constraint on BS location as follows:
\begin{equation}
\label{eq.P2.ad}
\begin{aligned}\notag
\text{\textrm{(P3)}}\ \underset{\{P_{k}\}, \bc}{\text{min}} \sum_k P_{k}\ \ \text{s.t.}\ \ &P_{k} \ge  \beta_k |\bc - \bx_k|^{\nu_k}, &|\bc-\ba_l|\le r_l,
\end{aligned}
\end{equation}
where $k=1...N,\ l=1..L$; the additional constraints $|\bc-\ba_l|\le r_l$ account for physical limitations or preferences, as discussed above, for given $\ba_l, r_l$.
An optimal BS location under these extra constraints can be characterized as follows.

\begin{thm}
\label{thm.c*.loc}
When (i) $\nu_k\ge 2$, or/and (ii) $\nu_k\ge 1$ and $\bc^* \neq \bx_k$, the optimal BS location for the problem (P3) can be expressed as a convex combination of user and constraint locations:
    \begin{equation}
    \label{eq.c*.loc.1}
\bc^* = \sum_{k=1}^{N+L} \theta_k \bx_k,
    \end{equation}
where $\bx_{N+l}=\ba_l,\ l=1...L$,
    \bal
    \label{eq.c*.loc.2}
    &\theta_k = \Theta^{-1}\nu_k \beta_k |\bc^*-\bx_k|^{\nu_k-2},\ k=1...N,\\
    \label{eq.c*.loc.3}
    &\theta_{N+l} =2\Theta^{-1}\mu_l,\ l=1...L,\\
    &\Theta = \sum_{k=1}^N \beta_k \nu_k|\bc^*-\bx_k|^{\nu_k-2}+ 2\sum_{l=1}^L \mu_l,
    \eal
and dual variables $\mu_l\ge 0$ are found from
\bal
\label{eq.c*.loc.4}
\mu_l(|\bc^*-\ba_l|-r_l)=0,
\eal
subject to $|\bc^*-\ba_l| \le r_l$. Signaling with the least per-user power is optimal: $P_k^* = \beta_k|\bc^*-\bx_k|^{\nu_k}$.
\end{thm}

\section{Conclusion}
In this paper, the problem of determining an optimal BS location for a given set of users was formulated as a convex optimization problem to minimize the total BS power subject to QoS constraints. Its globally-optimal solution was expressed as a convex combination of user locations. Based on this, a number of closed-form solutions were obtained, which revealed the impact of system and user parameters, propagation pathloss, as well as the overall system geometry. The symmetry in the user set was shown to make the optimal BS location independent of pathloss exponent, which is not true for asymmetric sets. These results provide insights unavailable from numerical algorithms, and allow one to develop design guidelines for more complicated systems. 


\section{Appendix: Proof of Theorem \ref{thm.c*}}

Since the problem (P2) is convex and the strong duality holds (since Slater condition is satisfied), its KKT conditions are sufficient for optimality \cite{Boyd2004}. Its Lagrangian is
\bal
L(P_k, \bc) = \sum_k P_{k}+ \sum_k \lambda_k (\beta_k |\bc - \bx_k|^{\nu_k} - P_{k}),
\eal
where $\gl_k \ge 0$ are Lagrange multipliers responsible for the power constraints. First, we consider the non-singular case, when $\bc^* \neq \bx_k\ \forall k$, and deal with the singular case later on. In the non-singular case, the KKT conditions take the following form
\bal
\label{sta-con}
&\sum_k \gl_k \beta_k \nu_k (\bc -\bx_k) |\bc-\bx_k|^{\nu_k-2}=0,\ 1-\lambda_k=0,\\
\label{com-slac}
&\lambda_k (\beta_k |\bc - \bx_k|^{\nu_k} - P_{k})=0,\\
\label{primal-con}
& P_{k} \ge \beta_k |\bc - \bx_k|^{\nu_k} , \lambda_k \ge 0,
\eal
where \eqref{sta-con} are the stationary conditions, \eqref{com-slac} are the complementary slackness conditions, and \eqref{primal-con} are primal and dual feasibility conditions. The 1st condition in \eqref{sta-con} was obtained from
\begin{equation}
\label{eq.dx}
\partial{|\bx|^{\nu}}/\partial{\bx} = \nu \bx|\bx|^{\nu-2}
\end{equation}
if $\bx \neq 0$, which always holds in the non-singular case. The 2nd condition in \eqref{sta-con} implies $\lambda_k=1$ so that, from \eqref{com-slac}, $P_{k}= \beta_k |\bc - \bx_k|^{\nu_k}$, i.e. transmitting with the least required power for each user is optimal. Combining this with the 1st condition in \eqref{sta-con} results, after some manipulations, in \eqref{eq.c*}.

The singular case, when $\bc^* = \bx_k$ for some $k$, is more involved as, in this case, \eqref{eq.dx} and hence 1st condition in \eqref{sta-con} do not hold (since $\bx =0$ and $|\bx|$ is not differentiable at $\bx=0$). One way to deal with this difficulty is to consider a regularized version of (P2), see \cite{Kalantari-20} for details.


\end{document}